\documentclass[11pt]{article}
\usepackage[margin=1in]{geometry}
\usepackage{amsmath,amsthm,amssymb}
\usepackage{algorithm,algorithmic,soul}
\usepackage[fleqn,tbtags]{mathtools}
\usepackage{bbm}
\usepackage{color}
\usepackage{tikz}
\usepackage{subcaption}
\allowdisplaybreaks

\usepackage[round]{natbib} 

\newtheorem{lemma}{Lemma}[section]

\newtheorem{theorem}[lemma]{Theorem}
\newtheorem{corollary}[lemma]{Corollary}
\newtheorem{definition}[lemma]{Definition}
\newtheorem{example1}[lemma]{Example}
\newtheorem{rem1}[lemma]{Remark}

\newtheorem{alg1}[lemma]{Algorithm}
\newtheorem{me1}[lemma]{Mechanism}
\newenvironment{remark}{\begin{rem1}\rm}{\end{rem1}}
\newenvironment{example}{\begin{example1}\rm}{\end{example1}}

\numberwithin{equation}{section}
\newcommand{\R}{\mathbb{R}}

\DeclareMathOperator{\cl}{cl}

\DeclareMathOperator{\interior}{int}
\DeclareMathOperator{\dom}{dom}

\DeclareMathOperator{\recc}{recc}

\renewcommand{\O}{\Omega}
\renewcommand{\o}{\omega}
\newcommand{\G}{\mathcal{G}}
\newcommand{\F}{\mathcal{F}}

\newcommand{\N}{\mathbb{N}}

\renewcommand{\Pr}{\mathbb{P}}

\newcommand{\of}[1]{\ensuremath{\left( #1 \right)}}
\newcommand{\cb}[1]{\ensuremath{ \left\{ #1 \right\} }}
\newcommand{\sqb}[1]{\ensuremath{ \left[ #1 \right] }}
\newcommand{\norm}[1]{\ensuremath{ \left\Vert #1 \right\Vert }}

\def\prehp(#1,#2){\ensuremath{  #1 \cdot #2 }}

\newcommand{\E}{\mathbb{E}}

\newcommand{\T}{\top}

\usepackage{amsopn}
\usepackage[capitalize,nameinlink]{cleveref}[0.19]

\begin{document}

\title{On the Separability of Vector-Valued Risk Measures}
\author{\c{C}a\u{g}{\i}n Ararat\thanks{Bilkent University, Department of Industrial Engineering, Ankara, Turkey, cararat@bilkent.edu.tr.}
	\and
Zachary Feinstein\thanks{Stevens Institute of Technology, School of Business, Hoboken, NJ, USA, zfeinste@stevens.edu.}
}
\date{\today}
\maketitle

\begin{abstract}
Risk measures for random vectors have been considered in multi-asset markets with transaction costs and financial networks in the literature. While the theory of set-valued risk measures provide an axiomatic framework for assigning to a random vector its set of all capital requirements or allocation vectors, the actual decision-making process requires an additional rule to select from this set. In this paper, we define vector-valued risk measures by an analogous list of axioms and show that, in the convex and lower semicontinuous case, such functionals always ignore the dependence structures of the input random vectors. We also show that set-valued risk measures do not have this issue as long as they do not reduce to a vector-valued functional. Finally, we demonstrate that our results also generalize to the conditional setting. 
These results imply that convex vector-valued risk measures are not suitable for defining capital allocation rules for a wide range of financial applications including systemic risk measures.\\~\\
{\bf Key words:} vector-valued risk measure, capital allocation rule, set-valued risk measure, systemic risk measure, duality\\
{\bf MSC Codes:} 26E25, 46A20, 46N10, 91G45, 91G70
\end{abstract}

\section{Introduction}\label{sec:intro}

Measuring the risks of random vectors were initially motivated by their relevance for multi-asset markets with proportional transaction costs. In such markets, uncertain financial positions are denoted in the physical units of the underlying assets, hence they are modeled as random vectors. As introduced in \citet{jmt} in the coherent case and further developed in \citet{hh}, \citet{hhr} in the convex case, a set-valued risk measure assigns a set of deterministic capital requirement vectors to the financial position that make it acceptable. We refer the reader to \citet{bentahar}, \citet{cascos} for alternative approaches for defining set-valued risk measures and \citet{zachsurvey} for a  comparison of these approaches.

Later, a similar risk measurement problem showed up in the setting of financial networks, where random vectors are used for modeling random shocks that affect the ability of the member institutions to meet their obligations to their creditors. To that end, systemic risk measures were introduced in \citet{chen2013axiomatic} as real-valued multivariate functionals in an axiomatic framework. Due to the dimension reduction in mapping a random vector to a scalar, these risk measures do not yield a clear way to allocate capital to the member institutions. In \citet{feinstein2017measures}, systemic risk measures are defined as set-valued risk measures that assign to a random shock the set of all capital allocation vectors that yield acceptable aggregate outcomes for society. The scalar systemic risk measures studied in \citet{biagini2019unified} are closely related to the set-valued ones in \citet{feinstein2017measures} (or their generalizations) through certain scalarization procedures.

While set-valued risk measures can be used for risk quantification purposes in both applications discussed above, a further selection procedure is often necessary to determine a capital requirement or allocation rule that takes values in the set-valued risk measure. For instance, choosing a minimizer of the total capital over the value of a set-valued systemic risk measure would yield a so-called \emph{efficient allocation rule}, see \citet[Definition~3.3]{feinstein2017measures}. In \citet{biagini2020fairness}, \emph{systemic optimal allocations} are defined based on the dual representations of the systemic risk measures in \citet{biagini2019unified}. In general, these allocation rules are vector-valued functionals applied to random vectors, which is essentially a call for building a theory of ``vector-valued risk measures".

Herein, we adopt an axiomatic approach for defining vector-valued risk measures. Specifically, using the componentwise order for vectors, we define vector-valued risk measures as a generalization of monetary risk measures for univariate random variables (see, e.g., \citet[Chapter 4]{fs:sf}). In our definition, in addition to monotonicity, the usual cash-additivity condition (also called translativity) is relaxed to a much weaker condition, referred to as \emph{marginal domination property} (see \Cref{sec:vector}), which also covers cash-subadditivity as a special case. A trivial example of a vector-valued risk measure is the vector of univariate risk measures applied to the components of a random vector. In this case, the risk measure does not consider the dependence between the components of the random vector. Our main theorem is a negative result which states that convex (and lower semicontinuous) vector-valued risk measures do not exist beyond this trivial case.

The result has two important implications for applications. First, it implies that it is impossible to define non-trivial and cash-subadditive capital allocation rules for systemic risk measures without relaxing one of three fundamental properties: monotonicity, convexity, or lower semicontinuity. This is not the case for set-valued risk measures; indeed, in the literature, there are known examples, such as the \emph{market risk measures} studied in \citet{sv-avar}, \citet{svdrm}, that do not fall into this category. We also provide a geometric characterization of a class of convex set-valued risk measures that trivialize into a marginal-based functional.

The second implication is related to dynamic set-valued risk measures, which were introduced in \citet{feinstien-qfin,feinstein-fs}. One benefit of using set-valued functionals in the dynamic setting is that \emph{multiportfolio time-consistency} appears as a natural concept of time-consistency and is formulated using compositions of set-valued functions. Recent literature has shown that this concept also appears in the form of a \emph{dynamic programming principle} for some multivariate dynamic optimization problems when formulated using set-valued functions. In \citet{karnam2017dynamic}, 
\citet{feinstein-scalar}, \citet{kovacova-ts}, such problems have been explored through the lens of a \emph{moving scalarization}, i.e., non-constant scalarization weights over time are needed in order to recover a time-consistent strategy. The conditional version of our main result implies that a vector-valued dynamic risk measure would \emph{not} be suitable as a (nonlinear) dynamic allocation rule as no dependencies between the components of the random vector can be modeled.

The rest of this paper is organized as follows. In \Cref{sec:defn}, we recall the basic properties of monetary (scalar and univariate) risk measures. In \Cref{sec:vector}, we introduce vector-valued risk measures and prove the main result. In \Cref{sec:implication}, we discuss the implications of this result for set-valued risk measures and both static and dynamic capital allocation rules.

\section{Definitions}\label{sec:defn}

We fix the setup and notation for the rest of the paper. Let $N\in\N:=\{1,2,\ldots\}$. For $x=(x_1,\ldots,x_N)^{\T}, y=(y_1,\ldots,y_N)^{\T}\in \R^N$, we write $x\leq y$ if $x_i\leq y_i$ for each $i\in\{1,\ldots,N\}$. Then, we define $\R^N_+:=\{x\in\R^N\mid 0\leq x\}$ and $\R^N_-:=\{x\in\R^N\mid x\leq 0\}$. We fix an arbitrary norm $|\cdot|$ on $\R^N$ and denote by $e_i$ the $i^{\text{th}}$ standard unit vector in $\R^N$ for each $i\in\{1,\ldots,N\}$. 
For nonempty sets $C,D\subseteq \R^N$  and $\lambda\in\R$, we define $C+D:=\{x+y\mid x\in C, y\in D\}$ and $\lambda C:=\{\lambda x\mid x\in C\}$ in the Minkowski sense; for $x\in \R^N$, we simply write $x+D:=\{x\}+D$. For a set $C\subseteq \R^N$, we denote by $\cl C$ and $\interior C$ its closure and interior, respectively, in the usual topology on $\R^N$. We define the support function of $C$ by $\sigma_C(w):=\sup_{x\in C}w^{\T}x$ for every $w\in \R^N$. We denote by $\recc C:=\{y\in \R^N\mid \forall \lambda>0\colon \lambda y + C\subseteq C\}$ the recession cone of $C$, and by $C^\circ:=\{w\in \R^N\mid \forall x\in C\colon w^{\T}x\leq 0\}$ the polar cone of $C$. For a function $f\colon \R^N\to[-\infty,+\infty]$, we define its effective domain by $\dom f:=\{x\in\R^N\mid f(x)<+\infty\}$.

To introduce the probabilistic setup, let us fix a probability space $(\O,\F,\Pr)$. For $A\in\F$, we define its indicator function by $1_A(\o):=1$ if $\o\in A$ and $1_A(\o):=0$ if $\o\in A^c:=\O\setminus A$. We denote by $L^0(\R^N)$ the space of all $\F$-measurable $N$-dimensional random vectors $X=(X_1,\ldots,X_N)^{\T}$ that are identified up to $\Pr$-almost sure equality. Hence, (in)equalities between random vectors are understood in this sense. For each $p\in[1,+\infty)$, we denote by $L^p(\R^N)$ the space of all $X\in L^0(\R^N)$ with $\E[|X|^p]<+\infty$, where $\E[\cdot]$ denotes the expectation operator with respect to $\Pr$. We denote by $L^\infty(\R^N)$ the space of all $X\in L^0(\R^N)$ for which $\norm{X}_\infty:=\inf\{c>0\mid \Pr\{|X|\leq c\}=1\}<+\infty$. 
For a nonempty Borel set $D\subseteq\R^N$, we write $L^p(D):=\{X\in L^p(\R^N)\mid \Pr\{X\in D\}=1\}$ for $p\in\{0\}\cup[1,+\infty]$.

\subsection{Monetary Risk Measures}

We collect some commonly used properties for risk measures in the next definition. We refer the reader to \citet[Chapter~4]{fs:sf} for a detailed treatment of convex risk measures on $L^\infty(\R)$.

\begin{definition}\label{defn:riskmsr}
For a real-valued functional $\rho\colon L^\infty(\R)\to\R$, consider the following properties:
\begin{itemize}
\item \textbf{Monotonicity}: $X\geq Y$ implies $\rho(X)\leq \rho(Y)$ for every $X,Y\in L^\infty(\R)$.
\item \textbf{Cash-subadditivity:} $\rho(X+m) \leq \rho(X)-m$ for every $X\in L^\infty(\R)$, $m\in\R$.
\item \textbf{Cash-additivity:} $\rho(X+m) = \rho(X)-m$ for every $X\in L^\infty(\R)$, $m\in\R$.
\item \textbf{Cash-preserving property:} $\rho(m)=-m$ for every $m\in\R$.
\item \textbf{Convexity:} $\rho(\lambda X+(1-\lambda)Y)\leq \lambda \rho(X)+(1-\lambda)\rho(Y)$ for every $X,Y\in L^\infty(\R)$, $\lambda\in(0,1)$.
\item \textbf{Lower semicontinuity:} $\rho$ is lower semicontinuous with respect to the weak* topology on $L^\infty(\R)$.
\end{itemize}
The functional $\rho$ is called a \textbf{monetary risk measure} if it is monotone and cash-additive.
\end{definition}

\begin{remark}
Herein we consider only finite-valued functionals $\rho$. All subsequent results hold also for proper functionals $\rho$, i.e., $\rho(X) > -\infty$ for every $X \in L^\infty(\R)$ and $\rho(Y) < +\infty$ for some $Y \in L^\infty(\R)$. 
Due to the choice of domain, for a monetary risk measure $\rho$, this is equivalent to \emph{finiteness at zero}, i.e., $\rho(0) \in \R$; often a stronger, \emph{normalization}, property is imposed so that $\rho(0) = 0$.
\end{remark}

\section{Vector-Valued Risk Measures}\label{sec:vector}

We start by generalizing the properties in \Cref{defn:riskmsr} to the vector-valued setting using the componentwise order $\leq$ on $\R^N$. This yields the next definition.

\begin{definition}\label{defn:vector}
For a vector-valued functional $r=(r_1,\ldots,r_N)^{\T}\colon L^\infty(\R^N)\to\R^N$, consider the following properties:
\begin{itemize}
\item \textbf{Monotonicity:} $X\geq Y$ implies $r(X)\leq r(Y)$ for every $X,Y\in L^\infty(\R^N)$.
\item \textbf{Marginal domination property:} For every $i\in\{1,\ldots,N\}$, there exists a function $f_i\colon \R\to\R$ such that $r_i(m)\leq f_i(m_i)$ for every $m\in\R^N$.
\item \textbf{Cash-subadditivity:} $r(X+m)\leq r(X)-m$ for every $X\in L^\infty(\R^N)$, $m\in\R^N$.
\item \textbf{Cash-additivity:} $r(X+m)=r(X)-m$ for every $X\in L^\infty(\R^N)$, $m\in\R^N$.
\item \textbf{Cash-preserving property:} $r(m)=-m$ for every $m\in\R^N$.
\item \textbf{Convexity:} $r(\lambda X+(1-\lambda)Y)\leq \lambda r(X)+(1-\lambda)r(Y)$ for every $X,Y\in L^\infty(\R^N)$, $\lambda\in(0,1)$.
\item \textbf{Positive homogeneity:} $r(\lambda X)=\lambda r(X)$ for every $X\in L^\infty(\R^N)$, $\lambda\geq 0$.
\item \textbf{Lower semicontinuity:} $r_i$ is lower semicontinuous with respect to the weak* topology on $L^\infty(\R^N)$ for every $i\in\{1,\ldots,N\}$.
\item \textbf{Separability:} There exist real-valued functionals $\bar{r}_1,\ldots,\bar{r}_N\colon L^\infty(\R)\to\R$ such that
\[
r(X)=(\bar{r}_1(X_1),\ldots,\bar{r}_N(X_N))^{\T},\quad X\in L^\infty(\R^N).
\]
\end{itemize}
The functional $r$ is called a \textbf{vector-valued risk measure} if it is monotone and has the marginal domination property. A vector-valued risk measure is called \textbf{coherent} if it is convex and positively homogeneous.
\end{definition}

\begin{remark}\label{rem:dom}
In \Cref{defn:vector}, monotonicity, cash-(sub)additivty, cash-preserving property, convexity, and lower semicontinuity are defined as the vectorial extensions of their univariate counterparts in \Cref{defn:riskmsr}. We introduce the marginal domination property as a weaker condition that is implied by cash-(sub)additivity and cash-preserving property. Indeed, when $r$ is cash-subadditive, taking $f_i(m_i):=r_i(0)-m_i$ for each $m_i\in\R$ and $i\in\{1,\ldots,N\}$ yields the marginal domination property. Similarly, when $r$ is cash-preserving, taking $f_i(m_i):=-m_i$ for each $m_i\in\R$ and $i\in\{1,\ldots,N\}$ yields the marginal domination property.
\end{remark}

The next theorem is the main result of the paper. It states that a convex vector-valued risk measure, under a mild continuity assumption, necessarily is separable.

\begin{theorem}\label{thm:implementable}
Let $r\colon L^\infty(\R^N)\to\R^N$ be a vector-valued risk measure that is convex and lower semicontinuous. Then, $r$ is also separable.
\end{theorem}

\begin{proof}
Let $i\in\{1,\ldots,N\}$. The assumptions on $r$ ensure that $r_i\colon L^\infty(\R^N)\to\R$ is a proper, convex, and lower semicontinuous functional. Hence, by Fenchel-Moreau theorem \citep[Theorem~2.3.3]{zalinescu}, it coincides with its biconjugate, i.e.,
\begin{equation}\label{eq:biconj}
r_i(X)=\sup_{U\in L^1(\R^N)}\of{\E[U^{\T}X]-r_i^\ast(U)},\quad X\in L^\infty(\R^N),
\end{equation}
where $r_i^\ast(U)\coloneqq \sup_{X\in L^\infty(\R^N)}(\E[U^{\T}X]-r_i(X))$ for every $U\in L^1(\R^N)$.
Let $U\in L^1(\R^N)$. Note that we have
\[
r_i^\ast(U)\geq \sup_{X\in L^\infty(\R^N_+)}\of{\E[U^{\T}X]-r_i(X)}\geq \sup_{X\in L^\infty(\R^N_+)}\E[U^{\T}X]-r_i(0),
\]
where the first inequality is by the definition of $r_i^\ast(U)$ and the second inequality is by monotonicity. Being the support function of the cone $L^\infty(\R^N_+)$, the last supremum is equal to $0$ if $U\in L^1(\R^N_-)$ and to $+\infty$ otherwise. Hence, $r_i^\ast(U)=+\infty$ if $U\notin L^1(\R^N_-)$. On the other hand,
\begin{align*}
r_i^\ast(U)&\geq \sup_{m\in\R^N}\of{\E[U]^{\T}m-r_i(m)}\\
&\geq \sup_{m\in\R^N}\of{\E[U]^{\T}m-f_i(m_i)}\\
&=\sup_{m_i\in\R}(\E[U_i]m_i-f_i(m_i))+\sum_{j\in\{1,\ldots,N\}\setminus\{i\}}\sup_{m_j\in\R}\E[U_j]m_j\\
&=f_i^\ast(\E[U_i])+\sum_{j\in\{1,\ldots,N\}\setminus\{i\}}\sup_{m_j\in\R}\E[U_j]m_j,
\end{align*}
where the first inequality is by the definition of $r_i^\ast(U)$, the second one is by the marginal domination property, and $f_i^\ast(v):=\sup_{r\in\R}(vr-f_i(r))$ for every $v\in\R$. Note that, for each $j\in\{1,\ldots,N\}\setminus\{i\}$, we have 
\[
\sup_{m_j\in\R}\E[U_j]m_j=\begin{cases}0 & \text{if }\E[U_j]=0,\\ +\infty & \text{else}.
\end{cases}
\]
Hence, $r_i^\ast(U)=+\infty$ if $\E[U_i]\notin \dom f_i^\ast$ or $\E[U_j]\neq 0$ for some $j\in\{1,\ldots,N\}\setminus\{i\}$.
Consequently, we may refine \eqref{eq:biconj} as
\begin{equation}\label{eq:biconj2}
r_i(X)=\sup_{U\in \mathcal{U}_i}\of{\E[U^{\T}X]-r_i^\ast(U)},\quad X\in L^\infty(\R^N),
\end{equation}
where $\mathcal{U}_i\coloneqq \{U\in L^1(\R^N_-)\mid \E[U_i]\in\dom f_i^\ast,\ \forall j\in\{1,\ldots,N\}\setminus\{i\}\colon \E[U_j]=0\}$. Note that every $U\in \mathcal{U}_i$ necessarily has $U_j=0$ for each $j\in\{1,\ldots,N\}\setminus\{i\}$ as a nonpositive random variable with zero expectation is zero almost surely. Therefore, \eqref{eq:biconj2} reduces to
\[
r_i(X)=\bar{r}_i(X_i)\coloneqq \sup_{\substack{U_i\in L^1(\R_-)\colon\\ \E[U_i]\in\dom f_i^\ast}}\of{\E[U_iX_i]-r_i^\ast(U_ie_i)},\quad X\in L^\infty(\R^N).
\]
This shows that $r$ is separable.
\end{proof}

\begin{remark}\label{rem:Lp}
Although we formulate our definitions and \Cref{thm:implementable} within the $L^\infty(\R^N)$ framework, the result can also be shown for a functional $r\colon L^p(\R^N)\to \R^N$ that satisfies monotonicity, marginal domination property, convexity, and lower semicontinuity with respect to the strong topology on $L^p(\R^N)$, where $p\in[1,+\infty)$. The proof would follow the same lines of arguments as above except that the dual elements $U$ are from the space $L^q(\R^N)$, where $q\in (1,+\infty]$ is such that $\frac{1}{p}+\frac{1}{q}=1$.
\end{remark}

\begin{remark}\label{rem:molchanov}
It is worth noting that \Cref{thm:implementable} works under the marginal domination property introduced in \Cref{defn:vector}, which is much weaker than cash-subadditivity. In \citet[Theorem~A.1]{molchanov:marginal}, a similar conclusion is obtained for vector-valued sublinear expectations, which can be identified as coherent risk measures $r\colon L^\infty(\R^N)\to\R^N$ that are cash-preserving, which are in turn cash-additive. Hence, our result relaxes the latter condition by the marginal domination property and removes the positive homogeneity assumption.
\end{remark}

The next corollary states \Cref{thm:implementable} for the cash-subadditive case.

\begin{corollary}
Let $r\colon L^\infty(\R^N)\to\R^N$ be a monotone, cash-(sub)additive functional that is convex and lower semicontinuous. Then, $r$ is also separable.
\end{corollary}

\begin{proof}
This follows directly from \Cref{rem:dom} and \Cref{thm:implementable}.
\end{proof}

\begin{remark}\label{rem:monotonicity}
\cref{thm:implementable} holds when all inequalities within \cref{defn:vector} are defined component-wise, i.e., with respect to the positive orthant ordering cone $\R^N_+$. 
As discussed previously, we are motivated in this work by the systemic risk measures of, e.g., \citet{feinstein2017measures}, in which each institution corresponds to a different component of the vector; for that setting, the positive orthant is the natural ordering cone to avoid implicit subsidies across firms.
If a different ordering cone $C \supsetneq \R^N_+$ were chosen, then the separability result no longer holds as presented and lacks a clear financial interpretation.
\end{remark}

\section{Implications of the Main Result}\label{sec:implication}

\subsection{Set-Valued Risk Measures and Capital Allocation Rules}\label{sec:svrm}

\Cref{thm:implementable} states that, under a mild continuity assumption, convex vector-valued risk measures only exist in the trivial coordinate-based form given by separability. In particular, such vector-valued risk measures do not take into account the dependence structure between the components of the input random vector. To allow for the use of dependence structures, one can use set-valued risk measures (see, e.g., \citet{hh}, \citet{hhr}) as an alternative. We recall their definition for convenience and propose possible definitions of separability for them. For convenience, let $\mathcal{J}(\R^N_+)$ denote the set of all closed convex sets $C\subseteq\R^N$ such that $\recc C=\R^N_+$.

\begin{definition}\label{defn:set-valued}
For a set-valued functional $R\colon L^\infty(\R^N)\rightrightarrows \R^N$ consider the following properties:
\begin{itemize}
\item \textbf{Monotonicity:} $X\geq Y$ implies $R(X)\supseteq R(Y)$ for every $X,Y\in L^\infty(\R^N)$.
\item \textbf{Cash-additivity:} $R(X+m)= R(X)-m$ for every $X\in L^\infty(\R^N)$, $m\in \R^N$.
\item \textbf{Finiteness at zero:} $R(0)\notin\{\emptyset,\R^N\}$.
\item \textbf{Convexity:} $R(\lambda X+(1-\lambda) Y)\supseteq \lambda R(X)+(1-\lambda)R(Y)$ for every $X,Y\in L^\infty(\R^N)$, $\lambda\in(0,1)$.
\item \textbf{Closedness:} The graph $\{(X,m)\in L^\infty(\R^N)\times\R^N\mid m\in R(X)\}$ is closed with respect to the product of the weak* topology on $L^\infty(\R^N)$ and the usual topology on $\R^N$.
\item \textbf{Vector-basedness:} There exist a lower semicontinuous 
functional $r\colon  L^\infty(\R^N) \to \R^N$ and a set $C\in\mathcal{J}(\R^N_+)$ such that
\[
R(X)=r(X)+C,\quad X\in L^\infty(\R^N).
\]
\item \textbf{Separability:} There exist set-valued functionals $\bar{R}_1,\ldots,\bar{R}_N\colon L^\infty(\R) \rightrightarrows \R$ with nonempty compact values and a set $C\in \mathcal{J}(\R^N_+)$ such that 
$\min \bar{R}_1,\ldots,\min\bar{R}_N\colon L^\infty(\R)\to \R$ are lower semicontinuous and 
\[
R(X)=\bar{R}_1(X_1)\times\ldots\times\bar{R}_N(X_N)+C,\quad X\in L^\infty(\R^N).
\]
\item \textbf{Strong separability:} There exist lower semicontinuous 
functionals $\bar{r}_1,\ldots,\bar{r}_N\colon L^\infty(\R)\to \R$ and a set $C\in \mathcal{J}(\R^N_+)$ such that
\[
R(X)= (\bar{r}_1(X_1),\ldots,\bar{r}_N(X_N))^{\T}+C, \quad X\in L^\infty(\R^N).
\]
\end{itemize}
The functional $R$ is called a \textbf{set-valued risk measure} if it satisfies monotonicity, cash-additivity, and finiteness at zero.
\end{definition}

\begin{remark}\label{rem:upperset}
For a closed convex set-valued risk measure $R$, it is easy to verify that $R(X)\notin\{\emptyset,\R^N\}$ is a closed convex set and $R(X)=R(X)+\R^N_+$ for every $X\in L^\infty(\R^N)$. 
\end{remark}

Vector-basedness is the property that the set-valued functional has a ``point plus a fixed set" structure. Separability and strong separability are two possible extensions of the similar property for vector-valued functionals. As a consequence of \Cref{thm:implementable}, the next theorem shows that having a ``point plus a fixed set" structure is equivalent to both of these separability properties for convex set-valued risk measures.

\begin{theorem}\label{prop:set-valued}
Let $R\colon L^\infty(\R^N)\rightrightarrows\R^N$ be a convex set-valued risk measure. Then, vector-basedness, separability, and strong separability are equivalent for $R$. In this case, $R$ is also closed.
\end{theorem}

The proof of \Cref{prop:set-valued} relies on the following ``cancellation lemma".

\begin{lemma}\label{lem:recc}
Let $C\in \mathcal{J}(\R^N_+)$ and $x,y\in\R^N$. Then, $x+C\supseteq y+C$ implies $x\leq y$.
\end{lemma}

\begin{proof}
By \citet[Corollary~14.2.1]{rockafellar}, we have $\dom \sigma_C\subseteq \cl \dom \sigma_C = (\recc C)^\circ =\R^N_-$. Moreover, by \citet[Theorem~6.3]{rockafellar}, $\interior \dom\sigma_C=\interior \cl \dom \sigma_C = \interior \R^N_-=:\R^N_{--}$. For each $w\in \R^N_{--}$, we have
\[
w^{\T}x+\sigma_{C}(w)=\sigma_{x+C}(w)\geq \sigma_{y+C}(w)=w^{\T}y+\sigma_C(w),
\]
which implies $w^{\T}x\geq w^{\T}y$ since $w\in \R^N_{--}\subseteq \dom\sigma_C$. Hence, $x\leq y$.
\end{proof}

\begin{proof}[Proof of \Cref{prop:set-valued}]
Clearly, strong separability implies both vector-basedness and separability.

Suppose that $R$ is vector-based with $r, C$ as in the property. Let $X,Y\in L^\infty(\R^N)$, $m\in\R^N$, $\lambda \in (0,1)$. If $X\geq Y$, then the monotonicity of $R$ implies that $r(X)+C\supseteq r(Y)+C$, which yields $r(X)\leq r(Y)$ by \Cref{lem:recc}. By the cash-additivity of $R$, we have $r(X+m)+C=r(X)-m+C$, which yields $r(X+m)=r(X)-m$ by \Cref{lem:recc}. By the convexity of $R$, we have 
\[
r(\lambda X+(1-\lambda)Y)+C\supseteq \lambda r(X)+(1-\lambda)Y+\lambda C+(1-\lambda)C,
\]
which yields $r(\lambda X+(1-\lambda)Y)\leq \lambda r(X)+(1-\lambda )r(Y)$ by the convexity of $C$ and \Cref{lem:recc}. Hence, $r$ is monotone, cash-additive, and convex. 
Since $r$ is also lower semicontinuous, by \Cref{thm:implementable}, $r$ is separable, which implies that $R$ is strongly separable. Moreover, let $w\in \R^N_-\setminus\{0\}$. Then, the function
\[
X\mapsto \sigma_{R(X)}(w)=-\sum_{i=1}^N (-w_i)r_i(X)+\sigma_C(w)
\]
is upper semicontinuous on $L^\infty(\R^N)$ as the sum of finitely many lower semicontinuous functions is lower semicontinuous. By \citet[Proposition~4.9, Proposition~4.23(b)]{setoptsurvey}, it follows that $R$ is closed.

Finally, suppose that $R$ is separable with $\bar{R}_1,\ldots,\bar{R}_N, C$ as in the property. For each $i\in\{1,\ldots,N\}$, let $\bar{r}_i:=\min \bar{R}_i$. Then, for each $X\in L^\infty(\R^N)$, we have
\begin{align*}
\sigma_{R(X)}(w)&=\sum_{i=1}^N w_i\min \bar{R}_i(X_i) + \sigma_C(w)\\
&=w^{\T}(\bar{r}_1(X_1),\ldots,\bar{r}_N(X_N))^{\T}+\sigma_C(w)\\
&=\sigma_{(\bar{r}_1(X_1),\ldots,\bar{r}_N(X_N))^{\T}+C}(w)
\end{align*}
for every $w\in \R^N_-\setminus \{0\}$, which shows that $R(X)=(\bar{r}_1(X_1),\ldots,\bar{r}_N(X_N))^{\T}+C$ as both sets are closed and convex. Hence, $R$ is strongly separable.
\end{proof}

Thanks to \Cref{prop:set-valued}, a convex and closed set-valued risk measure that does not collapse to a vector-valued functional does not ignore the dependence structure between the components, at least for some choices of the input vector $X\in L^\infty(\R^N)$. Such set-valued risk measures are constructed in \citet{sv-avar}, \citet{svdrm} under the names ``market extensions" and ``market risk measures" in multi-asset markets with either proportional or convex transaction costs. Vector-based convex and closed set-valued risk measures with $C=\R^N_+$ can also be written in the form $R(X)=\bar{R}_1(X_1)\times\ldots\times\bar{R}_1(X_N)$, where $\bar{R}_1,\ldots,\bar{R}_N\colon L^\infty(\R)\rightrightarrows\R$ are convex and closed set-valued functions. The set-valued entropic risk measure in \citet{svdrm}, on the other hand, is vector-based with a set $C\in\mathcal{J}(\R^N_+)$ that is different from $\R^N_+$; an example of such a risk measure is provided in \cref{ex:entropic} below.

\begin{remark}
For a convex set-valued risk measure $R$, an alternative continuity condition that is stronger than closedness is the following:
\begin{itemize}
\item \textbf{Convex upper continuity:} $R^{-1}(D):=\{X\in L^\infty(\R^N)\mid R(X)\cap D\neq \emptyset\}$ is closed with respect to the weak* topology on $L^\infty(\R^N)$ for every closed convex set $D\subseteq \R^N$ such that $D=D+\R^N_-$.
\end{itemize}
When $R$ is convex upper continuous, for every $w\in\R^N_-\setminus\{0\}$, the scalarization $X\mapsto \sigma_{R(X)}(w)$ is upper semicontinuous with respect to the weak* topology on $L^\infty(\R^N)$ by \citet[Proposition~A.1.3]{supermtg}. In this case, if there exist a vector-valued functional $r\colon L^\infty(\R^N)\to\R^N$ and a set $C\in\mathcal{J}(\R^N_+)$ with $\dom \sigma_C = \R^N_-$ such that $R(X)=r(X)+C$ for every $X\in L^\infty(\R^N)$, then it follows that $X\mapsto w^{\T}r(X)$ is upper semicontinuous for every $w\in\R^N_-\setminus\{0\}$, which implies that $r$ is lower semicontinuous. While this yields an alternative notion of vector-basedness where the lower semicontinuity of $r$ is a direct consequence of the convex upper continuity of $R$, as discussed in \Cref{ex:entropic} below, the assumption $\dom \sigma_C = \R^N_-$ is generally not satisfied by the set-valued entropic risk measure, which is known to be vector-based.
\end{remark}

\begin{example}\label{ex:entropic}
Let $K\subseteq\R^N$ be a closed convex set with $K=K+\R^N_+$ and suppose that $0$ is a boundary point of $K$. Let $\beta\in \R^N_{++}$. For each $i\in\{1,\ldots,N\}$ and $X_i\in L^\infty(\R)$, let $\bar{r}_i(X_i):=\frac{1}{\beta_i}\log(\E[e^{-\beta_i X_i}])$. The set-valued entropic risk measure (see \citet[Proposition~4.1]{svdrm}) is given by $R(X)=(\bar{r}_1(X_1),\ldots,\bar{r}_N(X_N))^{\T}+C$, $X\in L^\infty(\R^N)$, where
\[
C:=\cb{-\of{\frac{\log(1-\beta_1 x_1)}{\beta_1},\ldots,\frac{\log(1-\beta_N x_{N})}{\beta_N}}^{\T}\mid x\in K \cap \bigtimes_{i=1}^N \of{-\infty,\frac{1}{\beta_i}} }.
\]
Since $\bar{r}_1,\ldots,\bar{r}_N$ are univariate entropic risk measures, they are lower semicontinuous. Moreover, it is easy to verify that $C\in\mathcal{J}(\R^N_+)$. However, it is possible that $\dom \sigma_C\subsetneq \R^N_{-}$. Indeed, let us assume that $N=2$ and $K=\{x\in\R^2\mid x_1+x_2\geq 0\}$. In this case, for each $w\in \R^N_{-}$, the calculation of the support function simplifies to
\[
\sigma_C(w)= \sup\cb{-\frac{w_1}{\beta_1}\log(y_1)-\frac{w_2}{\beta_2}\log(y_2)\mid \frac{1-y_1}{\beta_1}+\frac{1-y_2}{\beta_2}\geq 0,\ y_1>0,\ y_2>0}
\]
and it can be verified that the functional inequality constraint in the above problem is satisfied with equality at optimality. Then, replacing this inequality with an equality and rewriting the problem yields
\[
\sigma_C(w)=\sup_{y_1>0}\of{-\frac{w_1}{\beta_1}\log(y_1)-\frac{w_2}{\beta_2}\log\of{1+\frac{\beta_2}{\beta_1}(1-y_1)}}.
\]
Hence, when $w=-e_1$, we immediately get $\sigma_C(-e_1)=\sup_{y_1>0}\frac{1}{\beta_1}\log(y_1)=+\infty$. Hence, $-e_1\notin\dom\sigma_C$ and we get $\dom\sigma_C=\R^2_{--}$ by symmetry.
\end{example}

We note that set-valued risk measures have also been used as systemic risk measures in financial networks as introduced in \citet{feinstein2017measures} (see also \citet{biagini2019unified}, \citet{ararat2020dual}). In this setting, for a random shock $X\in L^\infty(\R^N)$ that affects a network with $N$ banks, the set $R(X)$ consists of capital allocation vectors that yield an acceptable aggregate outcome. Then, a main issue is to find a plausible rule that tells how to allocate capital for all banks depending on the random shock. This is formalized in the next definition.

\begin{definition}\label{defn:CAR}
Let $R\colon L^\infty(\R^N)\rightrightarrows \R^N$ be a set-valued risk measure. A vector-valued functional $r\colon L^\infty(\R^N)\to\R^N$ is called a \textbf{capital allocation rule} for $R$ if $r(X)\in R(X)$ for every $X\in L^\infty(\R^N)$.
\end{definition}

\begin{corollary}\label{cor:acceptable}
Let $R\colon L^\infty(\R^N)\rightrightarrows \R^N$ be a set-valued risk measure. If $r\colon L^\infty(\R^N)\to\R^N$ is a capital allocation rule for $R$ that is a convex and lower semicontinuous vector-valued risk measure, then $r$ is separable.
\end{corollary}

\begin{proof}
This is an immediate consequence of \Cref{thm:implementable}.
\end{proof}

In view of \Cref{cor:acceptable}, there does \emph{not} exist a capital allocation rule that is a convex and lower semicontinuous vector-valued risk measure which depends on the copula associated with the random shock. This calls for relaxing the monotonicity or convexity requirements on the capital allocation rule, as most constructions in the literature do (e.g., via the dual representation as in~\citet{biagini2020fairness}). 

\subsection{Dynamic Risk Measures and Time-Consistency}

To discuss the implication of \Cref{thm:implementable} on dynamic time-consistent allocation rules, we first generalize this result to the conditional setting. To that end, let $\mathcal{G}$ be a sub-$\sigma$-algebra of $\mathcal{F}$. We denote by $L^\infty_\G(\R^N)$ the set of all $\G$-measurable random vectors in $L^\infty(\R^N)$. The next definition is the conditional analogue of \Cref{defn:vector}.

\begin{definition}\label{defn:conditional}
For a vector-valued functional $r=(r_1,\ldots,r_N)^{\T}\colon L^\infty(\R^N)\to L^\infty_{\mathcal{G}}(\R^N)$, consider the following properties:
\begin{itemize}
\item \textbf{Monotonicity:} $X\geq Y$ implies $r(X)\leq r(Y)$ for every $X,Y\in L^\infty(\R^N)$.
\item \textbf{Locality:} $r(1_A X+1_{A^c}Y)=1_A r(X)+1_{A^c}r(Y)$ for every $X,Y\in L^\infty(\R^N)$, $A\in\G$.
\item \textbf{Conditional marginal domination property:} For every $i\in\{1,\ldots,N\}$, there exists a functional $f_i\colon L^\infty_{\G}(\R)\to L^\infty_{\G}(\R)$ such that $r_i(m)\leq f_i(m_i)$ for every $m\in L^\infty_{\G}(\R^N)$.
\item \textbf{Conditional cash-subadditivity:} $r(X+m)\leq r(X)-m$ for every $X\in L^\infty(\R^N)$, $m\in L^\infty_{\G}(\R^N)$.
\item \textbf{Conditional cash-additivity:} $r(X+m)=r(X)-m$ for every $X\in L^\infty(\R^N)$, $m\in L^\infty_{\G}(\R^N)$.
\item \textbf{Conditional cash-preserving property:} $r(m)=-m$ for every $m\in L^\infty_{\G}(\R^N)$.
\item \textbf{Convexity:} $r(\lambda X+(1-\lambda)Y)\leq \lambda r(X)+(1-\lambda)r(Y)$ for every $X,Y\in L^\infty(\R^N)$, $\lambda\in(0,1)$.
\item \textbf{Fatou property:} For every bounded sequence $(X^n)_{n\in\N}$ in $L^\infty(\R^N)$ that converges to some $X\in L^\infty(\R^N)$ almost surely, it holds $r_i(X)\leq \liminf_{n\rightarrow\infty} r_i(X^n)$ for every $i\in\{1,\ldots,N\}$. 
\item \textbf{Conditional separability:} There exist functionals $\bar{r}_1,\ldots,\bar{r}_N\colon L^\infty(\R)\to L^\infty_{\G}(\R)$ such that
\[
r(X)=(\bar{r}_1(X_1),\ldots,\bar{r}_N(X_N))^{\T},\quad X\in L^\infty(\R^N).
\]
\end{itemize}
The functional $r$ is called a \textbf{vector-valued conditional risk measure} if it satisfies monotonicity, locality, and the conditional marginal domination property. 
\end{definition}

\begin{remark}\label{rem:locality}
Similar to the univariate case (see \citet[Exercise~11.1.2]{fs:sf}), it can be checked that a functional $r\colon L^\infty(\R^N)\to L^\infty_{\G}(\R^N)$ satisfying monotonicity, conditional cash-additivity, and $r(0)=0$ also satisfies locality. Since we assume the conditional marginal domination property in the definition of vector-valued conditional risk measure, which is weaker than conditional cash-additivity, we include locality as an additional property in the definition.
\end{remark}

We state the conditional version of \Cref{thm:implementable} next.

\begin{theorem}\label{thm:cond}
Let $r\colon L^\infty(\R^N)\to L^\infty_{\G}(\R^N)$ be a vector-valued conditional risk measure that is convex and has the Fatou property. Then, $r$ is also conditionally separable.
\end{theorem}

\begin{proof}
Let us define a vector-valued functional $\tilde{r}\colon L^\infty(\R^N)\to\R^N$ by
\[
\tilde{r}(X):= \E[r(X)],\quad X\in L^\infty(\R^N).
\]
Clearly, $\tilde{r}$ is monotone and convex. Fix $i\in\{1,\ldots,N\}$ and let $f_i\colon L^\infty_{\G}(\R)\to L^\infty_{\G}(\R)$ be such that $r_i(m)\leq f_i(m_i)$ for each $m\in L^\infty_{\G}(\R^N)$. Then, $\tilde{r}_i(m)=\E[r_i(m)]\leq \E[f_i(m_i)]=:\tilde{f}_i(m_i)$ for every $m\in \R^N$. Hence, $\tilde{r}$ has the conditional marginal domination property. We claim that $\tilde{r}$ is lower semicontinuous. Let $i\in\{1,\ldots,N\}$, $a,b\geq 0$ and define
\[
\mathcal{A}_i^{a,b}:=\{X\in L^\infty(\R^N)\mid \tilde{r}_i(X)\leq a,\ \norm{X}_\infty\leq b\}.
\]
Let $(X^n)_{n\in\N}$ be a sequence in $\mathcal{A}_i^{a,b}$ that converges to some $X\in L^1(\R^N)$ in $L^1(\R^N)$. Then, there exists a subsequence $(X^{n_k})_{k\in\N}$ that converges to $X$ almost surely. In particular, $\norm{X}_\infty\leq b$. Moreover, $r_i(X)\leq \liminf_{k\rightarrow\infty}r_i(X^{n_k})$ by the Fatou property of $r$. Then, by monotonoicity and Fatou's lemma for expectations,
\[
\tilde{r}_i(X)=\E[r_i(X)]\leq \E\sqb{\liminf_{k\rightarrow\infty}r_i(X^{n_k})}\leq \liminf_{k\rightarrow\infty}\E[r_i(X^{n_k})]=\liminf_{k\rightarrow\infty}\tilde{r}_i(X^{n_k})\leq a.
\]
Hence, $X\in \mathcal{A}_i^{a,b}$ so that $\mathcal{A}_i^{a,b}$ is closed in $L^1(\R^N)$. Then, as a consequence of Krein-\v{S}mulian theorem (see, e.g., \citet[Lemma~A.68]{fs:sf}), $\{X\in L^\infty(\R^N)\mid \tilde{r}_i(X)\leq a\}$ is weak* closed in $L^\infty(\R^N)$. Hence, $\tilde{r}$ is lower semicontinuous.

Therefore, by \Cref{thm:implementable}, $\tilde{r}$ is separable, i.e., there exist functionals $\hat{r}_1,\ldots,\hat{r}_N \colon L^\infty(\R) \to \R $ such that
\[
\tilde{r}(X)=(\hat{r}_1(X_1),\ldots,\hat{r}_N(X_N))^{\T},\quad X\in L^\infty(\R^N).
\]
We claim that $r$ is conditionally separable. To get a contradiction, suppose that there exist $i\in \{1,\ldots,N\}$ and $X, Y\in L^\infty(\R^N)$ such that $X_i=Y_i$ and $\Pr\{r_i(X)>r_i(Y)\}>0$. Let $A_i:=\{r_i(X)>r_i(Y)\}\in \G$. Note that $\max\{r_i(X),r_i(Y)\}=1_{A_i}r_i(X)+1_{A_i^c}r_i(Y)$. Then,
\begin{align*}
\E[\max\{r_i(X),r_i(Y)\}]&= \E[1_{A_i}r_i(X)+1_{A_i^c}r_i(Y)]\\
&=\E[r_i(1_{A_i} X+1_{A_i^c}Y)]=\tilde{r}_i(1_{A_i}X+1_{A_i^c}Y)\\
&=\hat{r}_i(1_{A_i}X_i+1_{A_i^c}Y_i)=\hat{r}_i(Y_i)=\tilde{r}_i(Y)=\E[r_i(Y)],
\end{align*}
where the passage to the second line is by locality. Recalling $\max\{r_i(X),r_i(Y)\}\geq r_i(Y)$, we conclude that $\max\{r_i(X),r_i(Y)\}=r_i(Y)$, which implies that $\Pr(A_i)=0$, a contradiction.
\end{proof}

In a dynamic framework, consider a filtration $(\F_t)_{t\in\mathbb{T}}$ on $(\O,\F)$, where $\mathbb{T}:=\{0,\ldots,T\}$ for some $T\in\N$ (or $\mathbb{T}:=[0,T]$ for some $T>0$) and $\F_T=\F$. For each $t\in \mathbb{T}$, let $r^t\colon L^\infty(\R^N)\to L^\infty_{\F_t}(\R^N)$ be a convex vector-valued conditional risk measure with the Fatou property. Similar to \Cref{defn:CAR}, one can ask $r^t$ to be a capital allocation rule for a corresponding set-valued conditional risk measure (see \citet[Definition~2.5]{feinstien-qfin}). Then, a natural property for the dynamic family $(r^t)_{t\in\mathbb{T}}$ is time-consistency in a certain sense, e.g., as having $r^s(X)=r^s(-r^t(X))$ for every $s<t$ in $\mathbb{T}$ and $X\in L^\infty(\R^N)$. However, \Cref{thm:cond} implies that $(r^t)_{t\in\mathbb{T}}$ would not be suitable as a dynamic time-consistent capital allocation rule due to conditional separability.

\section*{Acknowledgments}
The first author is partially supported by T\"{U}B{\.I}TAK (Scientific \& Technological Research Council of Turkey), Project No. 123F357.

\bibliographystyle{plainnat}
\bibliography{arxiv-SC}

\end{document}